\documentclass[runningheads,a4paper]{llncs}

\usepackage{graphicx}
\usepackage{amsmath}
\usepackage{paralist}
\usepackage{bm}
\usepackage{algorithmicx,algorithm}
\usepackage{algpseudocode}
\usepackage{url}
\usepackage{hyperref}
\usepackage{boxedminipage}
\usepackage{wrapfig}
\usepackage{ifthen}
\usepackage{color}
\usepackage{xcolor}
\usepackage{varwidth}
\usepackage{framed}

\spnewtheorem{observation}[obs]{Observation}{\bfseries}{\itshape}

\usepackage{thmtools}
\usepackage{thm-restate}
\newcommand{\ignore}[1]{}

\newcommand{\cA}{{\cal A}}

\newcommand{\cG}{\mathcal{G}}

\newcommand{\cN}{{\cal N}}

\newcommand{\eps}{\varepsilon}

\newcommand{\poly}{\mathrm{poly}}

\newcommand{\E}{\hbox{\bf E}}

\newcommand{\Pe}[2]{(#1)_{#2}}

\newcommand{\mG}{\mathbb{G}}

\newcommand{\Sec}[1]{\hyperref[sec:#1]{Section\,\ref*{sec:#1}}} %section
\newcommand{\Eqn}[1]{\hyperref[eq:#1]{(\ref*{eq:#1})}} %equation
\newcommand{\Fig}[1]{\hyperref[fig:#1]{Fig.\,\ref*{fig:#1}}} %figure
\newcommand{\Tab}[1]{\hyperref[tab:#1]{Tab.\,\ref*{tab:#1}}} %table
\newcommand{\Thm}[1]{\hyperref[thm:#1]{Theorem\,\ref*{thm:#1}}} %theorem
\newcommand{\Fact}[1]{\hyperref[fact:#1]{Fact\,\ref*{fact:#1}}} %fact
\newcommand{\Lem}[1]{\hyperref[lem:#1]{Lemma\,\ref*{lem:#1}}} %lemma
\newcommand{\Prop}[1]{\hyperref[prop:#1]{Prop.~\ref*{prop:#1}}} %property
\newcommand{\Cor}[1]{\hyperref[cor:#1]{Corollary~\ref*{cor:#1}}} %corollary
\newcommand{\Conj}[1]{\hyperref[conj:#1]{Conjecture~\ref*{conj:#1}}} %conjecture
\newcommand{\Def}[1]{\hyperref[def:#1]{Definition~\ref*{def:#1}}} %definition
\newcommand{\Alg}[1]{\hyperref[alg:#1]{Algorithm~\ref*{alg:#1}}} %algorithm
\newcommand{\Ex}[1]{\hyperref[ex:#1]{Ex.~\ref*{ex:#1}}} %example
\newcommand{\Clm}[1]{\hyperref[clm:#1]{Claim~\ref*{clm:#1}}} %example
\newcommand{\Obs}[1]{\hyperref[obs:#1]{Observation~\ref*{clm:#1}}} %example

\def\max{{\sf max}}

\newcommand{\crr}{\mathrm{Cr}}
\newcommand{\lp}[1]{\log_{\frac{1}{p}} #1}

\usepackage{amssymb}
\usepackage{graphicx}

\usepackage{url}
\urldef{\mailsa}\path|{alfred.hofmann, ursula.barth, ingrid.haas, frank.holzwarth,|
\urldef{\mailsb}\path|anna.kramer, leonie.kunz, christine.reiss, nicole.sator,|
\urldef{\mailsc}\path|erika.siebert-cole, peter.strasser, lncs}@springer.com|    
\newcommand{\keywords}[1]{\par\addvspace\baselineskip
\noindent\keywordname\enspace\ignorespaces#1}

\begin{document}

\mainmatter  % start of an individual contribution

\title{On Counting Cliques, Clique-covers and Independent sets in Random Graphs\footnote{ supported in part by NSF Grant CCF-1320814}}

\titlerunning{Counting  in Random Graphs}

\author{Kashyap Dixit%
%\thanks{}%
\and Martin F{\"u}rer}
\authorrunning{Kashyap Dixit \and Martin F{\"u}rer}

\institute{Pennsylvania State University\\
111 IST Building, University Park 16801, USA\\
kashyap@cse.psu.com\\
furer@cse.psu.edu\\
%\url{http://www.springer.com/lncs}
}


\toctitle{Counting  in Random Graphs}
\tocauthor{K.Dixit  {\em and} M.F{\"u}rer}
\maketitle

\begin{abstract}
We study the problem of counting the number of {\em isomorphic} copies of a given {\em template} graph, say $H$, in the input {\em base} graph, say $G$. In general, it is believed that polynomial time algorithms that solve this problem exactly are unlikely to exist. 
So, a lot of work has gone into designing efficient {\em approximation schemes}, especially, when $H$ is a perfect matching. 
%Recently, such schemes have been proposed for a general class of subgraph isomorphism counting problems by (F{\''u}rer and Kasiviswanathan, 2008).
 In this work, we present efficient approximation schemes to count $k$-Cliques, $k$-Independent sets and $k$-Clique covers in random graphs. \\
%We show that these embedding problems are \emph{almost} always approximable. The decision versions of these problems are fundamental and are well studied.  These problems are NP-hard. In particular, the {\em $k$-Clique-cover} problem is NP-hard for every $k\ge 3$. \\

We present {\em fully polynomial time randomized approximation schemes} (fpras) to count $k$-Cliques and $k$-Independent sets in a random graph on $n$ vertices when $k$ is at most $(1+o(1))\log n$, and $k$-Clique covers when $k$ is a constant. 
The problem of counting $k$-cliques and $k$-independent sets was an open problem in [Frieze and McDiarmid, 1997].
In other words, we have a fpras to evaluate the first $(1+o(1))\log n$ terms of the {\it clique polynomial} and the {\it independent set polynomial} of a random graph. 
[Grimmett and McDiarmid, 1975] present a simple greedy algorithm that {\em detects} a clique (independent set) of size $(1+o(1))\log_2 n$ in $G\in \cG(n,\frac{1}{2})$ with high probability. No algorithm is known to detect a clique or an independent set of larger size with non-vanishing probability. 
Furthermore, [Coja-Oghlan and Efthymiou, 2011] present some evidence that 
one cannot hope to easily improve a similar, almost 40 years old bound for sparse random graphs. Therefore, our results are unlikely to be easily improved. \\
%That is, it is unlikely to get better bounds.\\

We use a novel approach to obtain a recurrence corresponding to the variance of each estimator. Then we upper bound the variance using the corresponding recurrence. This leads us to obtain a polynomial upper bound on the critical ratio.
As an aside, we also obtain an alternate derivation of the closed form expression for the $k$-th moment of a binomial random variable using our techniques. The previous derivation [Knoblauch (2008)] was based on the moment generating function of a  binomial random variable.
\end{abstract}

\section{Introduction}
Given a {\em base graph} $G$ and a {\em template graph} $H$, {\em the subgraph isomorphism} problem is to decide whether an edge preserving injection $\phi$ between the vertices of $H$ and $G$ exists. 
That is, for every edge $\{u,v\}$ in $H$, $\{\phi(u),\phi(v)\}$ is an edge in $G$. Subgraph isomorphism is a generalization of several fundamental NP-complete problems, like Hamiltonian Path and Clique. 
The problem has applications in many areas, including cheminformatics~\cite{U76}, pattern discovery in databases~\cite{KK07}, bioinformatics~\cite{PCJ06} and social networks~\cite{SPRH06}. %and pattern matching~(\cite{}). 

Another widely studied related fundamental problem is that of counting the number of copies of $H$ in $G$. In general, this problem is \#P-complete (Valiant~\cite{V79}). The class \#P is defined as $\{f:$ 
There exists a non-deterministic polynomial time Turing machine $M$, such that on input $x$, the computation tree of $M$ has exactly $f(x)$ accepting leaves$\}$. The problems complete in this class are computationally quite difficult, 
since an oracle access to \#P complete problem would make it possible to solve any problem in the polynomial hierarchy in polynomial time (Toda~\cite{T91}). 

The $k$-Clique problem asks whether there exists a $k$-clique in the input graph $G$. A $k$-Clique is the complete graph on $k$ vertices. The $k$-Clique problem has numerous applications, particularly in bioinformatics and social networks~\cite{PCJ06,SPRH06}. 
Counting $k$-cliques in a web-graph has applications in social network analysis. In particular, this gives an estimate of the number of closed communities in the web-graphs. Therefore, fast algorithms for counting $k$-cliques in web-graphs give an insight to the evolution of Internet. 

The $k$-Clique cover problem asks for the existence of a perfect $k$-clique packing in $G$. More precisely, given base graph $G$ with $n$ vertices and template graph $H$ that is $n/k$ vertex disjoint and edge disjoint copies of $k$-cliques, does there exist an injective mapping from $H$ to $G$. 
The decision problem $k$-Clique Cover, that is $\{(G,k)\text{: There exists a disjoint cover of $G$ by $k$-cliques}\}$ is NP-complete on general graphs with clique number $3$~\cite{K72}.  The $k$-Clique cover problem has applications in the {\em orgy problem}~\cite{CHW82}:
 Given a group of people with affinities and aversion between them, is it possible to divide them into $k$  members each, such that every person in each group is compatible with every other person in the group. 
%Also, note that a special case of $k$-Clique cover problem is the $k$ dimensional matching problem, which is also NP-hard for $k\ge 3$. 
 Some of the scheduling problems can also be modeled as an orgy problem. 
We are given $n$ jobs of length $\le T$ seconds and $n/k$ machines. Also, for each job $j$, we are given a list of {\em conflicting jobs} which can not be scheduled with $j$ on the same machine. The problem is to schedule the jobs on the machines such that the total time to complete all the jobs is minimized. 

The clique-polynomial~\cite{HL94} of a graph $G=(V,E)$ is given by $1+\sum_{i=1}^{\omega(G)}c_i x^i$. Here, $c_i$ denotes the number of $i$-Cliques in $G$, $\omega(G)$ denotes the size of largest clique in $G$. The independent-set polynomial~\cite{HL94} of a graph is defined analogously. 
In general, computing the clique-polynomial and the independent set polynomial of a graph $G$ is $\#P$-complete. 

We consider template graphs which are vertex disjoint union of cliques. More specifically, we will be considering problems of counting cliques and clique covers. We note that our techniques can be extended to counting embeddings of template graphs which are disjoint union of cliques of possibly different sizes.
The counting version of the $k$-Clique problem is \#P-complete in general.  The counting version of the $k$-Clique cover problem is \#P-complete even for $k=2$ (Valiant~(\cite{V79})), where $H$ is a perfect matching. 

Note that the counting versions of the aforementioned problems are extremely hard even for the simple cases. So, we try to come up with {\em fully polynomial time approximation schemes} (abbreviated as fpras) for these problems that work well for {\em almost all} graphs. More precisely, 
fpras must run in time $\poly(n, \eps^{-1})$ and return an answer within a relative error of $(1\pm \eps)$ with high probability (i.e., probability tending to $1$ as $n\rightarrow \infty$) for graphs that are uniformly randomly sampled from $G\in \cG(n, p)$. Here, $\cG(n, p)$   denotes the class of 
graphs in which each edge occurs with probability $p$. Note that when $p=\frac{1}{2}$, each graph $G\in \cG(n,p)$ is equiprobable. Another commonly studied model is $\mathbb{G}(n,m)$ where each graph with $n$ vertices 
and $m$ edges is assigned the same probability, which is $\binom{N}{m}^{-1}$, where $N=\binom{n}{2}$. 

The theory of random graphs was initiated by Erd\H{o}s and R\'{e}nyi~\cite{ER60}. We work with the model $\cG(n,p)$ where we are given a fixed set of $n$ vertices and each of the $\binom{n}{2}$ edges is added with probability $p$. 

Our analysis also provides an alternate derivation of the closed form of the $k^{th}$ moment of a binomial random variable $X$ sampled from $\mathrm{Binomial}(n,p)$, which has been derived by Knoblauch~\cite{K08a} using {\em moment generating function}. We derive the same results using simple binomial equalities that we obtain using the binomial theorem.
\subsection{Our results}
In this work, we present new results for $k$-Clique and $k$-Clique cover counting problems  in random graphs. Our algorithm is based on the idea of Rasmussen's unbiased estimator for permanents~\cite{R94}. 
It has been widely used in the context of subgraph isomorphism counting problems~\cite{R97,FK05,FK08}. For counting $k$-cliques in the input random graph $G$, we embed a $k$-clique into $G$, doing so one vertex at a time chosen randomly. 
If the procedure succeeds, we compute the probability with which the clique is obtained in $G$ 
and output its inverse. As shown in~\cite{FK08}, this is an unbiased estimate of the number of cliques in $G$. We state the results below in \Thm{app-clique}.  
In this work, we generalize Rasmussen's approach~\cite{R94} to efficiently count $k$-cliques and $k$-clique covers in random graphs. As a corollary, we also get a fpras for counting $k$-independent sets in random graphs. 
Note that~\cite{CE11} indicates that our bounds is extremely difficult to be improved. 
%Therefore, \Thm{app-clique} almost completely answers an open question in the survey by Frieze and McDiarmid~\cite{FM97}. 
\begin{theorem}\label{thm:app-clique}
Let $H$ be a $k$-Clique, where $k=(1+o(1))\lp{n}$. Then, there exists a fpras for estimating the number of copies of $H$ in $G\in \cG(n,p)$ for constant $p$.
\end{theorem}

Note that counting $k$-cliques in $\cG(n,p)$ is equivalent to counting $k$-independent sets in $\cG(n,1-p)$. Since $p$ is a constant in our case, we have a fpras for counting $k$-Independent sets of a random graph.

\begin{theorem}\label{cor:app-ind-sets}
Let $H$ be a $k$-independent set, where $k=(1+o(1))\lp{n}$. Then, there exists a fpras for estimating the number of copies of $H$ in $G\in \cG(n,1-p)$ for constant $p$.
\end{theorem}

For counting $k$-clique cover, we embed one clique at a time, until the whole graph is covered by $k$-cliques. The {\em key} observation here is that after embedding a clique, the residual base graph still remains random with edge probability $p$. We obtain the following theorem for counting $k$-clique covers.
\begin{theorem}\label{thm:app-clique-covers}
Let $H$ be a $k$-clique cover, where $k=O(1)$.  Then, there exists a fpras for estimating the number of copies of $H$ in $G\in \cG(n,p)$ for constant $p$.
\end{theorem}

 Our estimators for counting cliques and clique-covers are given in \Alg{app-clique} and \Alg{app-clique-cover} respectively in \Sec{app-estimators}.  As a side result, we obtain an alternate derivation of $\E[X^k]$  for a binomial random variable $X$, for all $k\ge 0$. We note that this has already been obtained in~\cite{K08a} 
using the moment generating function for binomial random variable. 
%Automorphisms are edge preserving permutations on the set of vertices. The set of automorphisms forms a group under composition. For a $k$-Clique cover, the size of automorphsim group is $(k!)^{\frac{n}{k}}$. We denote this number by $a(k)$. 

\textbf{Outline of the paper}: In \Sec{app-prev}, we give some of the related work to set perspective for our work. 
To introduce our techniques to the reader, we give a new derivation for the closed form of  $k$-th moment for binomial random variables using these techniques \Sec{app-model-definitions}. 
We move on to describe estimators for counting $k$-cliques and $k$-clique covers in \Sec{app-estimators}. We analyze these estimators for counting $k$-cliques and $k$-clique covers for random graphs in \Sec{app-clique} and \Sec{app-ccover} respectively, which is the main contribution of this paper. 

\section{Related work}\label{sec:app-prev}
A lot of work has been done in finding and counting of cliques and independent sets in graphs. One of the earliest result in the theory of random graphs is about showing that the independence number and clique number of a random graph $G\in \cG(n,\frac{1}{2})$ is about $2\log_2 n$. 
Grimmett and McDiarmid~\cite{GM75} analyzed simple greedy algorithm constructs an inclusion-maximal independent set. They showed that it yields an independent set of size $(1+o(1))\log_2 n$. Coja-Oghlan and Efthymiou~\cite{CE11} show some evidence for why no better algorithm could be found over many years. 

Luby and Vigoda~\cite{LV97} have shown a fully polynomial time scheme for counting independent sets in the graphs with maximum degree $\Delta\le 4$, which was later improved by Weitz~\cite{W06} to $\Delta\le 5$. On the other hand,
 Dyer, Freize and Jerrum~\cite{DFJ02} have shown that no fpras exists for counting independent sets in graphs with $\Delta\ge 25$ unless NP=RP. They also show that 
the Markov Chain Monte Carlo technique is likely to fail if $\Delta\ge 6$. Chandrasekaran et.al.~\cite{CCGSS11} have obtained fpras for higher degree graphs with large girths. 

A major breakthrough in counting perfect matchings ($2$-clique covers) was a polynomial time algorithm for planar graphs due to Kasteleyn~\cite{K61}.   For a bipartite graph, it corresponds to calculating the permanent of a $\{0,1\}$ matrix. In the seminal  paper of Valiant~\cite{V79}, 
 it has been shown to be \#P-complete, even though the decision version of this problem is in P.
The noted work of Jerrum, Sinclair and Vigoda~\cite{JSV01} presents a fpras for counting perfect matchings in bipartite graphs. The problem of existence and counting of covers in random graphs $G\in\cG(n,p)$ was addressed in the seminal work of Johansson, Kahn and Vu~\cite{JKV08}. 
They show that given a subgraph $H$, the number of $H$-covers in a random graph $G\in \cG(n,p)$ is $e^{-O(n)}(n^{v-1}p^m)^{n/v}$ for large enough $n$ with probability at least $1-n^{-\Omega(1)}$. 
Here $v=|V(H)|$ and $m=|E(H)|$.  Various approaches for getting an unbiased estimator with small variance have been explored for counting perfect matchings in other graphs. 
Some of these are determinant based approaches~\cite{GG81,KKLLL93,C04,LP86}, Markov chain Monte Carlo (MCMC) algorithms~\cite{B86,JS89,JSV01,BSVV08} and search based on Rasmussen's techniques ~\cite{R94,R97,FK05,FK08}. 
Chien~\cite{C04}  gives an efficient fpras for counting perfect matchings in random graphs. MCMC algorithms are polynomial time algorithms for all bipartite graphs. 
%It is unclear to us whether both the techniques can be generalized to get unbiased estimators for counting $k$-Cliques and $k$-Clique covers. 
The estimators based on Rasmussen's approach (from~\cite{R94}) have also been proved to work well in random graphs, where they lead to simple, polynomial time approximation schemes. In this work, we generalize Rasmussen's approach to efficiently count $k$-Cliques and $k$-Clique covers in random graphs. 
As a corollary, we also get a fpras for counting $k$-Independent sets in random graphs.

 Rasmussen~\cite{R97} has given a fpras for counting cliques and independent sets in random graphs. But it is unclear how to extend that algorithm for counting $k$-cliques~\cite{FM97} or $k$-Independent sets in random graphs. We note here that 
F{\"u}rer and Kasivaswanathan~\cite{FK08} have used similar techniques to get fpras for a large class of subgraph isomorphism problems. A fundamental constraint in their analysis was that the template subgraphs triangle-free. Thus,
their analysis could not be extended directly to get fpras for $k$-clique, $k$-independent set and $k$-clique cover problems.

\section{ $k^{th}$ moment of a binomial random variable}\label{sec:app-model-definitions}
Consider the binomial random variable $X=\mathrm{binomial}(n,p)$. We are interested in finding the $k^{th}$ moment of $X$, i.e. we want to find $\E[X^k]$. In this section, we give the closed form expression for $\E[X^k]$. 
We evaluate using new equalities obtained from well known binomial theorem. Note that 
\begin{eqnarray*}
\E[X^k]=\sum_{i=0}^{n}i^k\binom{n}{i}p^i(1-p)^{n-i}
\end{eqnarray*}

We start with the most fundamental equality known as binomial theorem given below.
\begin{eqnarray}
\label{eq:app-binom}
(1+x)^n &=& \sum_{i=0}^{n}\binom{n}{i}x^i
\end{eqnarray} 

Suppose we differentiate \Eqn{app-binom} with respect to $x$ and multiply by $x$ subsequently, we get the following equation.
\begin{eqnarray}
\label{eq:app-1st}
nx(1+x)^{n-1} &=& \sum_{i=0}^{n}i\binom{n}{i}x^i
\end{eqnarray}

Note that substituting $x=\frac{p}{1-p}$ in \Eqn{app-1st} and multiplying by $(1-p)^n$, we get $np = \sum_{i=0}^{n}i\binom{n}{i}p^i(1-p)^{n-i}$, which is the first moment of $X$. Suppose we differentiate \Eqn{app-1st} w.r.t. 
$x$ again and multiply by $x$ subsequently, we get
\begin{eqnarray}
\label{eq:app-2nd}
 x(1+x)^{n-1}\Pe{n}{1}+x^2(1+x)^{n-2}\Pe{n}{2}&=&\sum_{i=0}^{n}i^2\binom{n}{i}x^i
\end{eqnarray}

The term $\Pe{n}{i}$ denotes the falling factorial $n\cdot (n-1)\cdot (n-2)\cdots (n-i+1)=\frac{n!}{(n-i)!}$. Again, substituting $x=\frac{p}{1-p}$ in \Eqn{app-2nd} and multiplying $(1-p)^n$, 
we get $\Pe{n}{1}p+\Pe{n}{2}p^2 = \sum_{i=0}^{n}i^2\binom{n}{i}p^i(1-p)^{n-i}=\E[X^2]$. 
The above calculations show an emerging pattern for higher moments, which \Lem{app-binom} illustrates. 
\begin{lemma}
\label{lem:app-binom}
\begin{eqnarray}
\label{eq:app-binomk}
g(x,k)=\sum_{i=0}^{n}i^k\binom{n}{i}x^i=\sum_{j=1}^{k}\lambda_{k,j}x^j(1+x)^{n-j}\Pe{n}{j}
\end{eqnarray}
Here $\lambda_{k,j}$ are the coefficients that depend on $k$ and $j$ but are independent of $n$. Here $0\le j\le k$ $\lambda_{k,0}=\lambda_{k,k+1}=0$.
\end{lemma}
\begin{proof}
We will prove the above lemma by induction. For $i=1$, this is true as shown in \Eqn{app-1st}. Suppose the lemma is true for $g(x,1),g(x,2),\dots,g(x,k)$. We prove that it holds for $g(x,k+1)$. Differentiating \Eqn{app-binomk} w.r.t. $x$ and subsequently multiplying with $x$ gives
\begin{align}
\label{eq:app-x}
\sum_{i=0}^{n}i^{k+1}\binom{n}{i}x^i &=\sum_{j=1}^{k}\lambda_{k,j}\Pe{n}{j}(jx^j(1+x)^{n-j}+(n-j)x^{j+1}(1+x)^{n-j-1})\nonumber\\
&= \sum_{j=1}^{k}\lambda_{k,j}jx^j(1+x)^{n-j}\Pe{n}{j}+\sum_{j=1}^{k}\lambda_{k,j}x^{j+1}(1+x)^{n-j-1}(n-j)\Pe{n}{j}\nonumber\\
%&=& \sum_{j=1}^{k}\lambda_{k,j}jx^j(1+x)^{n-j}P(n,j)+\sum_{j=1}^{k}\lambda_{k,j}x^{j+1}(1+x)^{n-j-1}(n-j)P(n,j)\nonumber\\
&= \sum_{j=1}^{k}\lambda_{k,j}jx^j(1+x)^{n-j}\Pe{n}{j}+\sum_{j=1}^{k}\lambda_{k,j}x^{j+1}(1+x)^{n-j-1}\Pe{n}{j+1}\nonumber\\
&= \sum_{j=1}^{k+1} (j\lambda_{k,j}+\lambda_{k,j-1})x^j(1+x)^{n-j}\Pe{n}{j}\nonumber\\
&= \sum_{j=1}^{k+1} \lambda_{k+1,j}x^j(1+x)^{n-j}\Pe{n}{j}
\end{align}

Note that the \Eqn{app-x} shows that $\sum_{i=0}^{n}i^{k+1}\binom{n}{i}x^i=\sum_{j=1}^{k+1} \lambda_{k+1,j}x^j(1+x)^{n-j}\Pe{n}{j}$ where $\lambda_{k+1,j}$ follows the recurrence relation 
$$\lambda_{k+1,j}=j\lambda_{k,j}+\lambda_{k,j-1}.$$ 
As given in~\cite{K08a}, Stirling numbers of second kind follow this recurrence. 
\begin{eqnarray}
\label{eq:app-coeff}
\lambda_{k,j}=\frac{1}{j!}\sum_{j=0}^{i}j^k\binom{i}{j}(-1)^j
\end{eqnarray}
\end{proof}
To get the $k^{th}$ moment, we simply substitute $x=\frac{p}{1-p}$ in \Eqn{app-binomk} and multiply by $(1-p)^n$. Hence we have the following theorem.

\begin{theorem}
\label{thm:app-binomk}
$$\E[X^k]= \sum_{j=1}^{k}\lambda_{k,j}p^j \Pe{n}{j}$$
where $\lambda_{k,j}$ are as given in \Eqn{app-coeff}.
\end{theorem}

\section{Estimators for counting $k$-cliques and $k$-clique covers in random graphs}\label{sec:app-estimators}
In this section, we formally describe our estimators. The estimator for counting cliques in given in \Alg{app-clique}. Note that it embeds the clique $\{v_1,\dots,v_k\}$ and outputs the inverse of probability of embedding it in this way into $G$.
The estimator embeds one vertex at a time until the whole clique is embedded. If the algorithm gets stuck, it outputs $0$. This process can be viewed as decomposing the clique into subgraphs $C_1,C_2,\dots, C_k$, where each $C_i$ is
 the subgraph induced by the $i^{th}$ numbered vertex $v_i$ and its lower numbered neighbors. It is denoted by $v_i$. 
%The {\em width} of a decomposition is defined as the size of largest subgraph embedded in any step~(\cite{FK08}). Note that the {\em width} of this graph decomposition is $k$ for each problem. 

We denote our randomized estimator by $\cA$ and let $X$ be the output estimate. To get an fpras, we need that $\E_\cA[X^2]/(\E_\cA[X])^2$, also called the {\em critical ratio}, is polynomially bounded. 
We will bound a related quantity called {\em critical ratio of averages} given by $\crr(X)=\E_\cG[\E_\cA[X^2]]/(\E_\cG[E_\cA[X]])^2$. Here, the outer expectation is over the graphs of $\cG(n,p)$  and the inner
 expectation is over the coin tosses of the estimator. Our focus in this work will be to get a bound on critical ratio of averages. As shown in \Prop{app-FK08}, this will also give a polynomial bound on the critical ratio itself.
The proof of \Prop{app-FK08} follows from \Cor{app-asymptotic} of \Thm{app-R00} from~\cite{R00}.

Consider any induced subgraph $H_v$ of $H$ with $v$ vertices. Let $e_{H}(v)=\max_{H_v\subseteq H}\{|E(H_v)|\}$ of edge  For stating the results, we need to define the following ratio for the template graph $H$.
$$\gamma=\gamma(H)=\max_{3\le v\le n}\{e_H(v)/(v-2)\}.$$

Note that $\gamma$ is closely related to the largest possible average degree of an induced subgraph of $H$. In our case, this is $(1+o(1))\log n$ for the case of counting cliques and $O(1)$ for counting clique covers. Let $C=C_H(G)$ 
denote the number of copies of $H$ in $G$.

\begin{theorem}[\cite{R00}]\label{thm:app-R00}
Let $H$ be a graph on $n$ vertices and $\gamma$ be as defined above. Let $p$ be a constant. Suppose that the following conditions hold: $p\cdot \binom{n}{2}\rightarrow \infty, \sqrt{n}(1-p)\rightarrow \infty$ and $np^{\gamma}/\Delta^4\rightarrow \infty$. Then, with high probability, 
a random graph $G \in  \mG\left(n, p\cdot\binom{n}{2}\right)$ has a spanning subgraph isomorphic to $H$. In general,
$C = C_H(G)$ satisfies 
$$\frac{\E[C^2]}{\E[C]^2}= 1 + o(1).$$
\end{theorem}
\emph{Remarks.} Note that \Thm{app-R00} holds for the spanning subgraphs of the random graphs. This assumption can easily be incorporated while embedding a single clique at any step. While embedding each clique, 
$H$ is considered to be the $n$ vertex graph which is the disjoint union of a clique and the isolated vertices in both the cases. 
Also, note that $np^\gamma/\Delta^4\rightarrow \infty$ since $\gamma$ and $\Delta$ are both bounded by $(1+o(1))\log n$. Therefore, all conditions of \Thm{app-R00} are satisfied in our case. So we get the following corollary in our case.

\begin{corollary}\label{cor:app-R00}
Let $G\in \mG(n,\Omega(n^2))$ and $H$ be one of the following graphs 
$$\text{(a) a clique of size $(1+o(1))\lp n$ or (b) a cover of cliques of constant size,}$$
Then $\frac{\E[C^2]}{\E[C]^2}= 1 + o(1)$, where $C$ denotes the number of copies of $H$ in $G$.
\end{corollary}

From the asymptotic equivalence between $\cG(n,p)$ and $\mG(n,m)$ (see e.g.~\cite{JLR00,R94}), we have the following corollary.

\begin{corollary}\label{cor:app-asymptotic}
Let $G\in \cG(n,p)$ and $H$ be one of the following graphs 
$$\text{(a) a clique of size $(1+o(1))\lp n$ or (b) a cover of cliques of constant size,}$$
Then $C\ge \E[C]/\omega$, where $\omega=\omega(n)$  be a real valued function that goes to $\infty$ as $n\rightarrow \infty$.
\end{corollary}
\Thm{app-R00} along with \Cor{app-R00} and \Cor{app-asymptotic} yield the following proposition. The proof is identical to the one given for a similar proposition in~\cite{FK08}, but we give it make the write-up self contained.

\begin{proposition}\label{prop:app-FK08}
Let $G\in \cG(n,p)$ and $H$ be one of the following graphs 
$$\text{(a) a clique of size $(1+o(1))\lp n$ or (b) a cover of cliques of constant size.}$$ Let X be the output of Algorithm Embeddings, and let $p$ be a constant. Then, for a random graph $G\in \cG(n,p)$ 
the critical ratio satisfies $\frac{\E[X^2]}{(\E[X])^2}\le \omega^3 \frac{\E_G[\E_\cA[X^2]]}{(\E_\cG[\E_\cA[X]])^2}$, where $\omega=\omega(n)$ such that $\omega\rightarrow \infty$ as $n\rightarrow \infty$.
\end{proposition}

\begin{proof}
For the unbiasted estimator $\cA$, we have $C=\E_\cA[X]$. Therefore, from \Cor{app-asymptotic}, we have that $C=\E_\cA[X]\le \E_\cG[\E_\cA[X]]/\omega$ with high probability.  Also, from Markov's inequality we have $\Pr[\E_\cA[X^2]> \omega \E_G[\E_\cA[X^2]]]\le 1/\omega$.
Therefore with probability at least $1-1/\omega$, we have $\E_\cA[X^2]\le \omega \E_G[\E_\cA[X^2]]$. Our result follows from these inequalities.
\end{proof}

In the rest of the paper, we focus on bounding the critical ratio of averages. The estimator for counting $k$-cliques is given in \Alg{app-clique}. It embeds one clique of size $k=(1+o(1))\lp n$ in $G$ and outputs the inverse of probability of embedding. This is done by the procedure \textproc{Embed-Clique}, 
which is called only once in this case.

\begin{algorithm}
\caption{Count-cliques($G,k$)}\label{alg:app-clique}
\begin{algorithmic}[1]	
\Procedure{Embed-Clique}{$G,k$}%\Comment{The g.c.d. of a and b}
\State $i\gets 0$ \Comment{$i$ denotes the number of nodes already embedded in $G$}
\State $v_1\gets \mathrm{ArbitraryNode}(G)$ \Comment{Arbitrarily assign a node from $G$ to $v_0$}
\While{$i < k$}%\Comment{We have the answer if r is 0}
\State $\cN_i\gets \mathrm{CommonNeighbors}(\{v_1,\dots,v_i\})$%\Comment{$\cN_i$ is the set of common neighbors of nodes $\{v_0,v_1,\cdots,v_i\}$}
\If {$\cN_i=\emptyset$} 
	\State $X\gets 0$ \Comment{Embedding algorithm has failed; so terminate}
\EndIf\State $X_i\gets |\cN_i|$
\State $v_{i+1}\gets \mathrm{RandomNode}(\cN_i)$\Comment{uniformly randomly assign a node from $\cN_i$ to $v_{i+1}$}
%Comment{uniformly at random}
\State $X\gets X\cdot X_i$
\State $i\gets i+1$
\EndWhile\label{euclidendwhile}
%\State $G_res\gets G\setminus \{v_0,v_1,\dots v_n\}$
\State \textbf{return} $X/(k!)$\Comment{Estimator outputs unbiased estimate of number of $k$-Cliques}\label{step:1last}
\EndProcedure
\end{algorithmic}
\end{algorithm}

The estimator for counting $k$-clique covers of $G$ is given in \Alg{app-clique-cover}. It uses the procedure \textproc{Embed-Clique} described in \Alg{app-clique} to embed each $k$-clique in the cover. This process is sequentially repeated until all the vertices are covered. 
In the end, it returns the inverse of probability of finding the cover, if successful. Note that this is the product of the probabilities of embedding the individual cliques in the cover. 
\begin{algorithm}
\caption{Count-clique-covers($G,k$)}\label{alg:app-clique-cover}
\begin{algorithmic}[1]
\State $G_{res}\gets G$
\State $a\gets (k!)^{\frac{n}{k}}\cdot(\frac{n}{k})!$\Comment{Size of the automorphism group of $k$-clique cover}
\State $X\gets 1$
\While{$G_{res}\neq \emptyset$}
\State $X\gets X\cdot  $\textproc{Embed-Clique($G_{res},k$)}
\If {\textproc{Embed-Clique($G_{res},k$) $= 0$}} 
	\State $X\gets 0$ \Comment{Embedding algorithm has failed; so terminate}
\EndIf
\State $G_{res}\gets G\setminus \{v_1,\dots v_k\}$\Comment{Remove the currently embedded clique $\{v_1,v_2,\dots, v_k\}$ from $G$ to get $G_{res}$}
\EndWhile
\State \textbf{return} $X/a$
\end{algorithmic}
\end{algorithm}

\section{Analysis of estimator for counting cliques and clique-covers in random graphs}
In this section, we show a polynomial bound on the critical ratio of averages for the estimators in \Alg{app-clique} and \Alg{app-clique-cover}. Note that from \Prop{app-FK08}, this is sufficient to bound the critical ratio of the estimator and hence get an fpras for 
counting $k$-cliques (for $k=(1+o(1))\log n$) and $k$-clique covers (for $k=O(1)$) in random graphs.

%we show how to count cliques and clique covers in a random graph using random embeddings. We prove \Thm{clique} in \Sec{clique}  and \Thm{clique-covers} in \Sec{ccover} respectively.
% We will do so by the random embedding as shown in \Alg{embed}

%\begin{algorithm}\label{alg:embed}
%\caption{Clique-Cover($G,k$)}
%\begin{algorithmic}[1]
%\SState Let $X\gets 1, G_f\gets G$ and $embed(0)\gets \emptyset$
%\BState {\em loop}:
%\SState $G'\gets G_f$
%\If (Random-Embedding($G_f, k$))
%\EndIf
%\BState {\em until} $\ell\le \frac{\log n}{\log\log n}$
%\Procedure{Random-Embedding}{$G, k$}
%\EndProcedure
%\end{algorithmic}
%\end{algorithm}

%%%%%%%%%%%%%%%%%%%%%%%%%%%%%%%%%%%%%%%%%%%%%%%%%%%%%%%%%%%%%%%%%%%%%%%%%%%%%%%%%
\subsection{Counting Cliques}\label{sec:app-clique}
In this section, we prove \Thm{app-clique}. In this case, the estimator embeds a single clique onto the base graph and outputs the inverse of probability of embedding the same. Let $X$, the random variable denoting the count,
 be the output of the estimator. The estimator selects first vertex in the graph arbitrarily and embeds 
one edge at a time until the whole clique is embedded. It outputs the inverse of probability of embedding if it goes through, else it outputs $0$. 

Let $X_j$ corresponds to the number of ways to embed vertex $j$ in the residual graph. Note that $X=X_{1}\cdot X_{2}\cdots X_{k}$. Now consider the term $\crr(X)=\E_\cG[\E_\cA[X^2]/\E_\cG[\E_\cA[X]]]^2$. 

To estimate the critical ratio of averages, we need the definition of $k$-nesting, denoted by $N(k, n, p)$, as follows.
\begin{definition}[$k$-nesting]\label{def:app-k-nesting}
A $k$-nesting is a function $N(k,n,p)$ that can be evaluated in the following recursive way. 
\begin{enumerate}[(i)]
\item The $2$-nesting is defined as  $$N(2,n,p)=n^2\left(\sum_{i=1}^{n-1}i^2\binom{n-1}{i}p^i(1-p)^{n-1-i}\right)$$
\item The $k$-nesting is defined as $$N(k,n,p)=n^2\left(\sum_{i=k-1}^{n-1}N(k-1,i ,p)\binom{n-1}{i}p^i(1-p)^{n-1-i}\right)$$
\end{enumerate}
\end{definition}

Note that the embedding of a $k$-clique can be thought of as embedding $i^{th}$-vertex to get an $i$-clique from $i-1$-clique for each $i\in \{1,2,\dots,k\}$. So, we have the following observation.

\begin{observation}\label{obs:app-variance} 
\begin{eqnarray}\label{eq:app-num-term}
%\E_\cG[\E_\cA[X_{n -ki+1}^2X_{n -ki+2}^2\cdots X_{n-(k-1)i}^2]]= N(k, n-ki, p)
%\E_\cG[\E_\cA[X_{k(i-1)+1}^2X_{k(i-1)+2}^2\cdots X_{k(i-1)+k}^2]]= N(k, n, p)
\E_\cG[\E_\cA[X_{1}^2X_{2}^2\cdots X_{k}^2]]= N(k, n, p)
\end{eqnarray}
\end{observation}

\Lem{app-num-term} shows the exact structure of $N(k,\ell,p)$,  which we use in getting the bound on the critical ratio. 
\begin{lemma}
\label{lem:app-num-term}
$$N(k,\ell, p)=\sum_{j=k}^{2k-1}\ell\Pe{\ell}{j}f_{k,j}(p)$$
Here $f_{k,j}(p)$ is a function in $k,j, p$ that is independent of $\ell$ with the following properties. 
\begin{enumerate}[(i)]
\item $f_{k,k-i}(p)=0$ for all $i\in \{1,\dots,k\}$ and $f_{k,2k+i}(p)=0$ for all $i\ge 0$.
\item $f_{k+1,j}(p)=p^{j-1}\left((j-1)f_{k,j-1}(p)+f_{k,j-2}(p)\right)$.
\end{enumerate}
\end{lemma}

\begin{proof}
We prove this by induction on $k$. For the base case, {\em i.e.} for $k=2$ this is 
\begin{eqnarray*}
N(2,\ell,p) &=& \ell^2\left(\sum_{i=1}^{\ell-1}i^2\binom{\ell-1}{i}p^i(1-p)^{\ell-1-i}\right)\\
&=& \ell^2 \Pe{\ell-1}{1}p+\ell^2\Pe{\ell-1}{2}p^2 \text{ (using \Eqn{app-binomk} with $k=2$)}\\
&=& \ell\Pe{\ell}{2}p^{\binom{2}{2}}+\ell\Pe{\ell}{3}p^{2\binom{2}{2}}
\end{eqnarray*}

Suppose the claim is true for $N(i,\ell,p)$ for $i=\{1,2,\dots,k\}$. We will show that the claim is true for $i=k+1$. From \Def{app-k-nesting} we have
\begin{align}\label{eq:app-num-poly}
&N(k+1,\ell, p)= \ell^2\sum_{m=k}^{\ell-1}N(k,m,p)\binom{\ell-1}{m}p^m(1-p)^{\ell-1-m}\nonumber \\
=&~ \ell^2\sum_{m=k}^{\ell-1}\sum_{j=k}^{2k-1}\left(m\Pe{m}{j}f_{k,j}(p)\right)\binom{\ell-1}{m}p^m(1-p)^{\ell-1-m}\nonumber \\
=& \sum_{j=k}^{2k-1}\sum_{m=k}^{\ell-1}\left(\ell^2\Pe{m}{j}\binom{\ell-1}{m}p^m(1-p)^{\ell-1-m} \right)f_{k,j}(p)\nonumber \text{ (interchanging the summations)}\\
=&\sum_{j=k}^{2k-1} (j\cdot \ell^2\Pe{\ell-1}{j}p^{j}+\ell^2\Pe{\ell-1}{j+1}p^{j+1})f_{k,j}(p)\nonumber \text{ (from \Lem{app-fact-mom}, Eqn.\Eqn{app-sub})} \\
%&=& \sum_{j=k}^{2k-1} (j\cdot \ell\Pe{\ell}{j+1}p^{j}+\ell\Pe{\ell}{j+2}p^{j+1})f_{k,j}(p)\nonumber  \text{ (using $\ell\Pe{\ell-1,i}=\Pe{\ell}{i+1}$)}\\
=&\sum_{i=k+1}^{2k+1} (j\cdot \ell\Pe{\ell}{i}p^{i-1}+\ell\Pe{\ell}{i+1}p^{i})f_{k,i-1}(p)\nonumber  \text{ (using $\ell\Pe{\ell-1}{i}=\Pe{\ell}{i+1}$), $j+1=i$, $f_{k,2k}(p)=0$)}\\
=& \sum_{i=k+1}^{2k+1}p^{i-1}((i-1)f_{k,i-1}(p)+f_{k,i-2}(p))\ell\Pe{\ell}{i}\nonumber \text{ (rearranging the terms and using $f_{k,k-1}(p)=0$)} \\
=& \sum_{i=k+1}^{2k+1}f_{k+1,i}(p)\ell\Pe{\ell}{i}\nonumber \text{ (rearranging the terms and using $f_{k,k-1}(p)=0$)} \\
\end{align}

\end{proof}
The following lemma is used in the proof of \Lem{app-num-term}.
\begin{lemma}\label{lem:app-fact-mom}
\begin{align}\label{eq:app-fact-mom}
\sum_{m=j}^{n}m\Pe{m}{j}\binom{n}{m}x^{m-j}=j\Pe{n}{j}(1+x)^{n-j}+\Pe{n}{j+1}x(1+x)^{n-j-1}
\end{align}
In particular, if we multiply \Eqn{app-fact-mom} by $x^j(1-p)^n$ and substitute $x=p/(1-p)$ we get
\begin{align}
\label{eq:app-sub}
\sum_{m=j}^{n}m\Pe{m}{j}\binom{n}{m}p^m(1-p)^{n-m}=j\Pe{n}{j}p^j+\Pe{n}{j+1}p^{j+1}
\end{align}
\end{lemma}
\begin{proof}
We prove the identity in \Eqn{app-fact-mom} using induction.For $j=0$ (base case) we need to show that
$\sum_{m=0}^nm\binom{n}{m}x^m=nx(1+x)^{n-1}$, which holds from \Eqn{app-1st}. For hypothesis, assume that \Eqn{app-fact-mom} holds for $j$. We prove that it also holds for $j+1$ as follows. 
Differentiating \Eqn{app-fact-mom} w.r.t. $x$ gives 
\begin{eqnarray}
\sum_{m=j+1}^{n}m(m-j)\Pe{m}{j}x^{m-j-1}&=& j(n-j)\Pe{n}{j}(1+x)^{n-j-1}\Pe{n}{j+1}(1+x)^{n-j-1}\nonumber\\&&+(n-j-1)\Pe{n}{j+1}x(1+x)^{n-j-2}\nonumber \\
\sum_{m=j+1}^{n}m\Pe{m}{j+1}x^{m-(j+1)}&=& j\Pe{n}{j+1}(1+x)^{n-j-1}+\Pe{n}{j+1}(1+x)^{n-j-1}\nonumber \\&&+\Pe{n}{j+2}x(1+x)^{n-j-2}\nonumber \text{ (using $n\Pe{n-1}{i}=\Pe{n}{i+1}$)} \\
\sum_{m=j+1}^{n}m\Pe{m}{j+1}x^{m-(j+1)}&=& (j+1)\Pe{n}{j+1}(1+x)^{n-j-1}+\Pe{n}{j+2}x(1+x)^{n-j-2}\nonumber
\end{eqnarray}

Hence the identity holds for $j+1$.
\end{proof}

The following lemma upper bounds $f_{k,k+i}(p)$ for $0\le i\le k-1$
\begin{lemma}\label{lem:app-f-upper-bound}
For $k\ge 2$ $f_{k,2k-i-1}(p)\le k^{2i}p^{\binom{k}{2}+\binom{k-i}{2}}$ where $0\le i\le k-1$. 
\end{lemma}

\begin{proof}
We will prove this claim using induction on $k$. Consider $k=2$ for the base case. From \Def{app-k-nesting}, we have $N(2,n,p)=n\Pe{n}{2}p+n\Pe{n}{3}p^2$. So, the claim holds. Now assume that the claim holds for all clique sizes up to $k-1$ 
Now, from \Eqn{app-num-poly}, we have the following recurrence relation. 
\begin{eqnarray}\label{eq:app-recurrence}
f_{k,i}(p)&=&p^{i-1}((i-1)f_{k-1,i-1}(p)+f_{k-1,i-2}(p))
\end{eqnarray}

First we prove for $i\ge 1$. Using \Eqn{app-recurrence}, we have 
\begin{align*}
f_{k,2k-i-1}(p)&=p^{2k-i-2}((2(k-1)-i)f_{k-1,2(k-1)-(i-1)-1}(p)+f_{k-1,2(k-1)-i-1}(p))\\
&\le p^{2k-i-2}\left((2(k-1)-i)(k-1)^{2(i-1)}p^{\binom{k-1}{2}+\binom{k-i}{2}}+(k-1)^{2i}p^{\binom{k-1}{2}+\binom{k-i-1}{2}}\right)\\
%&=& p^{2k-i-2+\binom{k-1}{2}+\binom{k-i-i}{2}}\left(\right)
&=(k-1)^{2i}p^{\binom{k}{2}+\binom{k-1}{2}}\left(1+p^{k-i-1}\left(\frac{2}{k-1}-\frac{i}{(k-1)^2}\right)\right)\\
&\le(k-1)^{2i} p^{\binom{k}{2}+\binom{k-1}{2}}\left(1+\frac{2}{k-1}\right)\\
&= \left((k-1)^{2i} +2(k-1)^{2i-1}\right)p^{\binom{k}{2}+\binom{k-1}{2}}\le k^{2i}p^{\binom{k}{2}+\binom{k-1}{2}} \text{ (for $i\ge 1$)}
\end{align*}

Now we show that $f_{k,2k-1}=p^{2\binom{k}{2}}$. From \Eqn{app-recurrence}, we have $f_{k,2k-1}(p)=p^{2(k-1)}((2k-2)f_{k-1, 2k-2}(p)+f_{k-1, 2k-3}(p))=p^{2(k-1)}f_{k-1,2k-3}(p)$ since $f_{k-1, 2k-2}(p)=0$. 
Applying the recurrence repeatedly, we get the desired relation.
\end{proof}

%Let $\crr(K_i)=\E_\cG[\E_\cA[X_{n -ki+1}^2X_{n -ki+2}^2\cdots X_{n-(k-1)i}^2]]/\E_\cG[\E_\cA[X_{n -ki+1}X_{n -ki+2}\cdots X_{n-(k-1)i}]]^2$. Note that this is $\frac{N(k,\ell,p)}{\left(\Pe{\ell}{k}p^{\binom{k}{2}}\right)^2}$ where $\ell=n-ki$.
Now we bound $\crr(X)$ which is the same as $\frac{N(k,n,p)}{\left(\Pe{n}{k}p^{\binom{k}{2}}\right)^2}$
We have
\begin{eqnarray}\label{eq:app-cri-rat}
%\crr(K_i)=\frac{N(k,\ell,p)}{\left(\Pe{\ell}{k}p^{\binom{k}{2}}\right)^2}=\sum_{j=k}^{2k-1}\frac{\ell\Pe{\ell}{i}f_{k,j}(p)}{\left(\Pe{\ell}{k}p^{\binom{k}{2}}\right)^2}&=&\sum_{i=0}^{k-1}\frac{\ell\Pe{\ell}{2k-i-1}f_{k,2k-i-1}(p)}{\left(\Pe{\ell}{k}p^{\binom{k}{2}}\right)^2}
\crr(X)=\frac{N(k,n,p)}{\left(\Pe{n}{k}p^{\binom{k}{2}}\right)^2}=\sum_{j=k}^{2k-1}\frac{n\Pe{n}{i}f_{k,j}(p)}{\left(\Pe{n}{k}p^{\binom{k}{2}}\right)^2}&=&\sum_{i=0}^{k-1}\frac{n\Pe{n}{2k-i-1}f_{k,2k-i-1}(p)}{\left(\Pe{n}{k}p^{\binom{k}{2}}\right)^2}
\end{eqnarray}

\Lem{app-clique} immediately proves \Thm{app-clique}.

\begin{lemma}\label{lem:app-clique}
For $k=(1+o(1))\lp n$,  $\crr(X)=\sum_{i=0}^{k-1}\frac{n\Pe{n}{2k-i-1}f_{k,2k-i-1}(p)}{\left(\Pe{n}{k}p^{\binom{k}{2}}\right)^2}$ is upper bounded by $\poly(n)$.
\end{lemma}

\begin{proof}
Consider the ratio $\frac{\ell\Pe{\ell}{2k-i-1}f_{k,2k-i-1}(p)}{\left(\Pe{\ell}{k}p^{\binom{k}{2}}\right)^2}$ for a fixed $i$. Here we have $\ell = n$. As we shall see, $\ell$ changes for the $k$-Clique cover. For $i=0$, this is $\frac{\ell\Pe{\ell}{2k-1}}{(\Pe{\ell}{k})^2}$ since $f_{k,2k-1}=p^{2\binom{k}{2}}$. Note that 
$\frac{\ell\Pe{\ell}{2k-1}}{(\Pe{\ell}{k})^2}\le 1$. Now we consider $i\ge 1$.
\begin{eqnarray}\label{eq:app-hi}
\frac{\ell\Pe{\ell}{2k-1-i}f_{k,2k-i-1}(p)}{\left(\Pe{\ell}{k}p^{\binom{k}{2}}\right)^2}&=&\left(\prod_{j=1}^{k-i-1}\frac{(\ell-(k-1)-j)}{(\ell-j)}\right)\left(\frac{f_{k,2k-i-1}(p)}{\prod_{r=1}^{i}(\ell-k+r)}\right)\frac{1}{p^{2\binom{k}{2}}}\nonumber \\
&\le& \left(\frac{\ell-k}{\ell-1}\right)^{k-i-1}\left(\frac{k^{2i}p^{\binom{k}{2}+\binom{k-i}{2}}}{(\ell-k+1)^i}\right)\frac{1}{p^{2\binom{k}{2}}}\nonumber  \\
&=& \left(\frac{\ell-k}{\ell-1}\right)^{k-i-1}\left(\frac{k^{2}}{\ell-k+1}\right)^i\frac{1}{p^{\binom{k}{2}-\binom{k-i}{2}}}\nonumber \\
&=& \left(\frac{\ell-k}{\ell-1}\right)^{k-i-1}\left(\frac{k^{2}}{\ell-k+1}\left(\frac{1}{p}\right)^{k-\left(\frac{i+1}{2}\right)}\right)^i = h(i)
\end{eqnarray}
The first inequality above uses \Lem{app-f-upper-bound}, $\frac{\ell-k}{\ell-1}\ge \frac{\ell-k-j}{\ell-1-j}$ and $\ell-k+j\ge \ell -k+1$ for all $1\le j\le k-1$. Note that $\left(\frac{\ell-j}{\ell-1}\right)^{k-i-1}\le 1$. So, we have 
$h(i)\le \left(\frac{k^{2}}{\ell-k+1}\left(\frac{1}{p}\right)^{k-\left(\frac{i+1}{2}\right)}\right)^i$, where $h(i)$ is as defined in \Eqn{app-hi}. Note that for $i=(1+o(1))\log n$, $\left(\frac{k^{2}}{\ell-k+1}\left(\frac{1}{p}\right)^{k-\left(\frac{i+1}{2}\right)}\right)^i$ 
is polynomially bounded for all $0\le i\le k-1$. Therefore $\crr(X)$ is polynomially bounded.
\end{proof}

\begin{lemma}\label{lem:app-poly}
For $k=(1+o(1))\lp{n}$, $h(i)=\left(\frac{k^{2}}{n-k+1}\left(\frac{1}{p}\right)^{k-\left(\frac{i+1}{2}\right)}\right)^i$ is polynomially bounded for all $0\le i\le k-1$.
\end{lemma}
\begin{proof}
First note that $$\left(\frac{k^{2}}{n-k+1}\left(\frac{1}{p}\right)^{k-\left(\frac{i+1}{2}\right)}\right)^i = \left(\frac{1}{p}\right)^{\left(2i\lp{k}-i\lp({n-k+1})+ki-\frac{i(i+1)}{2}\right)}$$

Let $g(i) = 2i\lp{k}-i\lp{(n-k+1)}+ki-\frac{i(i+1)}{2}$, where $k=(1+\eps_n)\lp{n}$. This function is maximized at the point where $\partial g(i)/\partial i =0$, which happens at
$i \approx2\lp\lp n+\eps_n\lp n $. At this point, $g(i)\approx 2\left(\lp\lp n\right)^2+\frac{\eps_n^2}{2}(\lp n)^2$. Note that $h(i)=\left(\frac{1}{p}\right)^{g(i)}$ is polynomially bounded only when $\eps_n=O(\frac{1}{\sqrt{\lp n}})$.

\end{proof}

%%%%%%%%%%%%%%%%%%%%%%%%%%%%%%%%%%%%%%%%%%%%%%%%%%%%%%%%%%%%%%%%%%%%%%%%%%%%%%%
\subsection{Clique cover counting}\label{sec:app-ccover}
As noted earlier in \Prop{app-FK08}, we focus on bounding the  {\em critical ratio of averages} given by $\crr(X)=\E_\cG[\E_\cA[X^2]]/(\E_\cG[E_\cA[X]])^2$ for \Alg{app-clique}. 

The estimator embeds one clique at a time, by selecting a vertex at random at first and then embedding each edge till $k$ vertices of the clique are embedded. A crucial observation is that the residual graph, after embedding a clique still remains random with edge probability $p$. 
Finally, the estimator sequentially embeds $n/k$ cliques to get the clique cover and outputs the inverse of probability of getting this clique cover, if the embedding procedure goes through, otherwise it outputs $0$. Note that this is the product of the inverse of the probabilities for embedding each clique. 
Let $K_{i}$ denote the random variable corresponding to the estimate of the number of embeddings of the $i^{th}$ clique in the residual graph, which is a random graph from $\cG(n-ki-k,p)$. Note that $K_i$ is independent from $K_j$ for $i\neq j$ and $X=K_1\cdot K_2\cdots K_{\frac{n}{k}}$.
Therefore we have the following equation.
\begin{eqnarray}\label{eq:app-den}
(\E_\cG[E_\cA[X]])^2=\prod_{i=1}^{\frac{n}{k}}(E[K_i])^2
\end{eqnarray}
Note that the equality follows from the fact that after embedding each $k$-clique, the residual graph still remains random with edge probability $p$. Now, we bound the numerator, {\em i.e.}, $\E_\cG[\E_\cA[X^2]]$.  
\begin{eqnarray}\label{eq:app-num}
\E_\cG[\E_\cA[X^2]]=\E_\cG[\E_\cA[K_1^2K_2^2\cdots K_{\frac{n}{k}}^2]]= \E_\cG[E_\cA[K_1^2]]\cdot \E_\cG[\E_\cA[K_2^2]]\cdots \E_\cG[E_\cA[K_{\frac{n}{k}}^2]]
\end{eqnarray}

Let $X_j$ corresponds to the number of ways to embed vertex $j$ in the residual graph. Note that $K_{i}=X_{ki-k+1}\cdot X_{ki-k+2}\cdots X_{ki}$. 

Now consider the term $\E_\cG[\E_{\cA}[K_{i}^2]]=\E_\cG[\E_\cA[X_{k(i-1)+1}^2X_{k(i-1)+2}^2\cdots X_{ki}^2]]$. 
Note that in this case, we have 
\begin{eqnarray}\label{eq:app-num-term}
\crr(K_{i})=\frac{\E_\cG[\E_\cA[X_{ki-k+1}^2X_{ki-k+2}^2\cdots X_{ki}^2]]}{\E_\cG[\E_\cA[X_{ki-k+1}X_{ki-k+2}\cdots X_{ki-k+k}]]^2}= \frac{N(k, n-ki+k, p)}{\left(\Pe{n-ki}{k}p^{\binom{k}{2}}\right)^2}
%\E_\cG[\E_\cA[X_{n -ki+1}^2X_{n -ki+2}^2\cdots X_{n-(k-1)i}^2]]= N(k, n, p)
\end{eqnarray}
We show in \Lem{app-cri-rat-int-bound} that $\frac{N(k,\ell,p)}{\left(\Pe{\ell}{k}p^{\binom{k}{2}}\right)^2}$ is bounded by $1+O\left(\frac{1}{n-ki+1}\right)$ for all $i\in\{1,2,\dots,\frac{n}{k}\}$, where $\ell=n-ki+k$. 
\begin{lemma}\label{lem:app-cri-rat-int-bound}
For large $\ell$, constant $k$  and constant $p$ we have $$\crr(K_{i})=\sum_{j=0}^{k-1}\frac{\ell\Pe{\ell}{2k-j-1}f_{k,2k-j-1}(p)}{\left(\Pe{\ell}{k}p^{\binom{k}{2}}\right)^2} \le 1+O\left(\frac{1}{\ell-k+1}\right)$$
\end{lemma}

\begin{proof}
Consider the ratio $\frac{\ell\Pe{\ell}{2k-j-1}f_{k,2k-j-1}(p)}{\left(\Pe{\ell}{k}p^{\binom{k}{2}}\right)^2}$ for a fixed $j$. For $j=0$, this is $\frac{\ell\Pe{\ell}{2k-1}}{(\Pe{\ell}{k})^2}$ since $f_{k,2k-1}=p^{2\binom{k}{2}}$. Note that 
$\frac{\ell\Pe{\ell}{2k-1}}{(\Pe{\ell}{k})^2}\le 1$. Now we consider $j\ge 1$. 
As shown in \Eqn{app-hi}, we have 
$$\frac{\ell\Pe{\ell}{2k-1-j}f_{k,2k-j-1}(p)}{\left(\Pe{\ell}{k}p^{\binom{k}{2}}\right)^2}\le h(j)=\left(\frac{\ell-k}{\ell-1}\right)^{k-j-1} \left(\frac{k^{2}}{\ell-k+1}\left(\frac{1}{p}\right)^{k-\left(\frac{j+1}{2}\right)}\right)^j $$
%\begin{eqnarray}\label{eq:hi}
%\frac{\ell\Pe{\ell}{2k-1-i}f_{k,2k-i-1}(p)}{\left(\Pe{\ell}{k}p^{\binom{k}{2}}\right)^2}&=&\left(\prod_{j=1}^{k-i-1}\frac{(\ell-(k-1)-j)}{(\ell-j)}\right)\left(\frac{f_{k,2k-i-1}(p)}{\prod_{r=1}^{i}(\ell-k+r)}\right)\frac{1}{p^{2\binom{k}{2}}}\nonumber \\
%&\le& \left(\frac{\ell-j}{\ell-1}\right)^{k-i-1}\left(\frac{k^{2i}p^{\binom{k}{2}+\binom{k-i}{2}}}{(\ell-k+1)^i}\right)\frac{1}{p^{2\binom{k}{2}}}\nonumber  \\
%&=& \left(\frac{\ell-j}{\ell-1}\right)^{k-i-1}\left(\frac{k^{2}}{\ell-k+1}\right)^i\frac{1}{p^{\binom{k}{2}-\binom{k-i}{2}}}\nonumber \\
%&=& \left(\frac{\ell-j}{\ell-1}\right)^{k-i-1}\left(\frac{k^{2}}{\ell-k+1}\left(\frac{1}{p}\right)^{k-\left(\frac{i+1}{2}\right)}\right)^i = h(i)
%\end{eqnarray}
%The first inequality above uses \Clm{f-upper-bound}, $\frac{\ell-k}{\ell-1}\ge \frac{\ell-k-j}{\ell-1-j}$ and $\ell-k+j\ge \ell -k+1$ for all $1\le j\le k-1$. 

To prove the lemma, we handle the cases of $j\le 2$ and $j\ge 2$  separately. First we handle the latter case. For $j\ge 2$, we prove that $h(j)\le \frac{1}{(k-2)(\ell-k+1)}$. In other words, we prove that $\log h(j)+\log(k-2)+\log(\ell-k+1)<0$ for constant $k$.

Let $y(j)=\log (h(j))=(k-1-j)(\log(\ell-k)-\log(\ell-1))+j(2\log k - log(\ell-k+1))+j\left(k-\frac{(j+1)}{2}\right)\log{\frac{1}{p}}$.  Consider the continuous function 
$y(x)=(k-1-x)(\log(\ell-k)-\log(\ell-1))+x(2\log k - log(\ell-x+1))+x\left(k-\frac{(x+1)}{2}\right)\log{\frac{1}{p}}$. Therefore we have 
\begin{eqnarray*}
y'(x)=\frac{\partial y(x)}{\partial x}&=&log (\ell -1)-(\log(\ell-k)+\log(\ell-k+1))+2\log k+\left(k-x+\frac{1}{2}\right)\log{\frac{1}{p}}\\
&\le& log (\ell -1)-2(\log(\ell-k))+2\log k+\left(k-x+\frac{1}{2}\right)\log{\frac{1}{p}}\\
\end{eqnarray*}

Observe that for large $\ell$ and for constant $k$, the term $-2(\log(\ell-k))$ dominates all the other terms, so $y'(x)<0$ for $1\le x\le k-1$. Therefore $y(x)$ is a decreasing function.  We analyze cases $j=1$ and $j\ge 2$ separately.
First we analyze latter case.  We prove that $y(2) \le\log\left(\frac{1}{(k-2)(\ell-k+1)}\right)$, which implies that $y(j)= \log\left(\frac{1}{(k-2)(\ell-k+1)}\right)$ for $2\le j\le k-1$. This proves that $h(i)=\frac{1}{(k-2)(\ell-k+1)}$ for $j\ge 2$ , eventually proving that
\begin{eqnarray}\label{eq:app-ige2}
\sum_{j=2}^{k-1}\frac{\ell\Pe{\ell}{2k-j-1}f_{k,2k-j-1}(p)}{\left(\Pe{\ell}{k}p^{\binom{k}{2}}\right)^2} \le \frac{1}{\ell-k+1}
\end{eqnarray} 

Consider the function $g(\ell)=y(2)+\log (k-2)+\log (\ell-k+1)$. So we have 
\begin{eqnarray*}
g(\ell)&=&(k-3)(\log(\ell-k)-\log(\ell-1))+2(2\log k - \log(\ell-k+1))\\&&+(2k-3)\log\frac{1}{p} +\log(\ell-k+1)+\log(k-2)\\
& \le  & 4\log k +(2k-3)\log\frac{1}{p}+\log(k-2)- \log(\ell-k+1)
\end{eqnarray*}

Note that for constant $k$, this is smaller than $0$ for large enough $\ell$. Therefore $g(\ell)<0$, hence the claim.

Now we do the analysis for $j=1$. We calculate $f_{k,2k-2}(p)$ using the recurrence. 
\begin{eqnarray*}
f_{k,2k-2}(p)&=& p^{2k-3}\left((2k-3)f_{k-1, 2(k-1)-1}(p)+f_{k-1, 2(k-1)-2}(p)\right)\\
&=& (2k-3)p^{2k-3+2\binom{k-1}{2}}+p^{2k-3}f_{k-1,2(k-1)-2} (p) \text{ (using $f_{k-1,2(k-1)-1}=p^{2\binom{k-1}{2}}$)}\\
&=&(2k-3) p^{2\binom{k}{2}-1} +p^{2k-3+2k-5}\left((2k-5)f_{k-2,2(k-2)-1}(p)+f_{k-2, 2(k-2)-2}(p)\right)\\
&=&(2k-3)p^{2\binom{k}{2}-1}+(2k-5)p^{2\binom{k}{2}-2}+p^{2k-3+2k-5}f_{k-2, 2(k-2)-2}(p)\\
\end{eqnarray*}

Going on as shown in the above equation, we get $f_{k,2k-2}(p)=\sum_{m=1}^{k-1}(2(k-m)-1)p^{2\binom{k}{2}-m}$. Therefore we have
\begin{eqnarray*}
\frac{\ell\Pe{\ell}{2k-2}}{\left(\Pe{\ell}{k}p^{\binom{k}{2}}\right)^2}&=&\frac{\ell\Pe{\ell}{2k-2}}{(\Pe{\ell}{k})^2}\left(\sum_{m=1}^{k-1}\frac{2(k-m)-1}{p^m}\right)\\
\end{eqnarray*}

Note that 
\begin{eqnarray}\label{eq:app-ie1}
\left(\sum_{m=1}^{k-1}\frac{2(k-m)-1}{p^m}\right)&=&\frac{1}{\frac{1}{p}-1}\left( 2\left(\frac{\frac{1}{p^k}-1}{{\frac{1}{p}-1}}\right)+\frac{1}{p^k}-\left(\frac{2k+2p-1}{p}\right)\right)\nonumber \\
&\le& \frac{C}{p^k} \text{ (for large enough constant $C$)}
\end{eqnarray}

Note that for constant $k$, $\frac{C}{p^k}=C'$ is a constant.
Therefore, using \Eqn{app-ige2} and \Eqn{app-ie1} we have 
\begin{eqnarray*}
\sum_{j=0}^{k-1}\frac{\ell\Pe{\ell}{2k-j-1}f_{k,2k-j-1}(p)}{\left(\Pe{\ell}{k}p^{\binom{k}{2}}\right)^2} \le 1+\left(\frac{C'+1}{\ell-k+1}\right)
\end{eqnarray*}
Hence the lemma.
\end{proof}

Note that \Lem{app-cri-rat-int-bound} shows that $\crr(K_i)=\E_\cG[\E_\cA[K_i^2]]/\E_\cG[\E_\cA[K]]^2=1+O\left(\frac{1}{n-ki+1}\right)$. Note that \Thm{app-clique-covers} follow from \Lem{app-cri-rat-int-bound} since 
$\prod_{i=1}^{\frac{n}{k}}\crr(K_i)=\poly(n)$ in this case.

\section{Conclusion and open problems}
In this work, we show the first fpras for counting $k$-cliques, where $k=(1+o(1))\lp n$ and $k$-clique covers (for constant $k$) in random graphs, using the unbiased estimators that are very simple to describe. Both problems are \#P-complete in general for the respective values of $k$. 
Getting a fpras for these problems over general graphs is a long standing open problem. Here are some specific open problems that we think are worth investigating.
\begin{enumerate}
\item The problem of counting clique is still open for counting cliques of size greater $(1+o(1))\lp n$. Solving this will resolve the open problem of Frieze and McDiarmid~(\cite{FM97}) completely, though, this is probably very hard to solve~\cite{CE11}.
\item Another specific problem to resolve here is to count clique covers of superconstant sized cliques. 
\item The determinant based estimators usually have smaller worst case running times in fpras (e.g.~\cite{C04}) for random graphs. It is unclear to us how to obtain any determinant based unbiased estimators for the clique and clique cover counting problems.
\end{enumerate}

\bibliographystyle{plain}
%\begin{spacing}{0.9}
\bibliography{references}

\begin{thebibliography}{10}

\bibitem{SPRH06}
Tom A.B.Snijders, Philippa~E. Pattison, , Garry~L. Robins, and Mark~S.
  Handcock.
\newblock New specifications for exponential random graph models.
\newblock {\em Socialogical Methodology}, 36(1):99--153, 2006.

\bibitem{BSVV08}
Ivona Bez{\'{a}}kov{\'{a}}, Daniel Stefankovic, Vijay~V. Vazirani, and Eric
  Vigoda.
\newblock Accelerating simulated annealing for the permanent and combinatorial
  counting problems.
\newblock {\em {SIAM} J. Comput.}, 37(5):1429--1454, 2008.

\bibitem{B86}
Andrei~Z. Broder.
\newblock How hard is to marry at random? (on the approximation of the
  permanent).
\newblock In {\em Proceedings of the 18th Annual {ACM} Symposium on Theory of
  Computing, May 28-30, 1986, Berkeley, California, {USA}}, pages 50--58, 1986.

\bibitem{CCGSS11}
Venkat Chandrasekaran, Misha Chertkov, David Gamarnik, Devavrat Shah, and
  Jinwoo Shin.
\newblock Counting independent sets using the bethe approximation.
\newblock {\em {SIAM} J. Discrete Math.}, 25(2):1012--1034, 2011.

\bibitem{C04}
Steve Chien.
\newblock A determinant-based algorithm for counting perfect matchings in a
  general graph.
\newblock In J.~Ian Munro, editor, {\em Proceedings of the Fifteenth Annual
  {ACM-SIAM} Symposium on Discrete Algorithms, {SODA} 2004, New Orleans,
  Louisiana, USA, January 11-14, 2004}, pages 728--735. {SIAM}, 2004.

\bibitem{CE11}
Amin Coja{-}Oghlan and Charilaos Efthymiou.
\newblock On independent sets in random graphs.
\newblock In Dana Randall, editor, {\em Proceedings of the Twenty-Second Annual
  {ACM-SIAM} Symposium on Discrete Algorithms, {SODA} 2011, San Francisco,
  California, USA, January 23-25, 2011}, pages 136--144. {SIAM}, 2011.

\bibitem{CHW82}
Gerard Cornuejols, David Hartvigsen, and William R.Pulleyblank.
\newblock Packing subgraphs in a graph.
\newblock In {\em Operations Research Letters}, volume~1, pages 139--143, 1982.

\bibitem{DFJ02}
Martin~E. Dyer, Alan~M. Frieze, and Mark Jerrum.
\newblock On counting independent sets in sparse graphs.
\newblock {\em {SIAM} J. Comput.}, 31(5):1527--1541, 2002.

\bibitem{ER60}
Paul Erd\H{o}s and Alfr\'{e}d R\'{e}nyi.
\newblock On the evolution of random graphs.
\newblock {\em Publ. Math. Inst. Hung. Acad. Sci.}, 5:17--61, 1960.

\bibitem{FM97}
Alan~M. Frieze and Colin McDiarmid.
\newblock Algorithmic theory of random graphs.
\newblock {\em Random Struct. Algorithms}, 10(1-2):5--42, 1997.

\bibitem{FK05}
Martin F{\"{u}}rer and Shiva~Prasad Kasiviswanathan.
\newblock Approximately counting perfect matchings in general graphs.
\newblock In Camil Demetrescu, Robert Sedgewick, and Roberto Tamassia, editors,
  {\em Proceedings of the Seventh Workshop on Algorithm Engineering and
  Experiments and the Second Workshop on Analytic Algorithmics and
  Combinatorics, {ALENEX} /ANALCO 2005, Vancouver, BC, Canada, 22 January
  2005}, pages 263--272. {SIAM}, 2005.

\bibitem{FK08}
Martin F{\"u}rer and Shiva~Prasad Kasiviswanathan.
\newblock Approximately counting embeddings into random graphs.
\newblock In Ashish Goel, Klaus Jansen, Jos{\'e} D.~P. Rolim, and Ronitt
  Rubinfeld, editors, {\em APPROX-RANDOM}, volume 5171 of {\em Lecture Notes in
  Computer Science}, pages 416--429. Springer, 2008.

\bibitem{GG81}
Chris~D. Godsil and Ivan Gutman.
\newblock On the matching polynomial of a graph.
\newblock pages 241--249, 1981.

\bibitem{GM75}
G.~R. Grimmett and C.~J.~H. McDiarmid.
\newblock On colouring random graphs.
\newblock 77:313--324, 1975.

\bibitem{HL94}
Cornelis Hoede and Xueliang Li.
\newblock Clique polynomials and independent set polynomials of graphs.
\newblock {\em Discrete Mathematics}, 125(1-3):219--228, 1994.

\bibitem{JLR00}
Svante Janson, Tomasz Luczak, and Andrzej Rucinski.
\newblock {\em Random Graphs}.
\newblock Wiley-Interscience, 2000.

\bibitem{JS89}
Mark Jerrum and Alistair Sinclair.
\newblock Approximating the permanent.
\newblock {\em {SIAM} J. Comput.}, 18(6):1149--1178, 1989.

\bibitem{JSV01}
Mark Jerrum, Alistair Sinclair, and Eric Vigoda.
\newblock A polynomial-time approximation algorithm for the permanent of a
  matrix with nonnegative entries.
\newblock {\em J. {ACM}}, 51(4):671--697, 2004.

\bibitem{JKV08}
Anders Johansson, Jeff Kahn, and Van~H. Vu.
\newblock Factors in random graphs.
\newblock {\em Random Struct. Algorithms}, 33(1):1--28, 2008.

\bibitem{KKLLL93}
Narendra Karmarkar, Richard~M. Karp, Richard~J. Lipton, L{\'{a}}szl{\'{o}}
  Lov{\'{a}}sz, and Michael Luby.
\newblock A monte-carlo algorithm for estimating the permanent.
\newblock {\em {SIAM} J. Comput.}, 22(2):284--293, 1993.

\bibitem{K72}
Richard~M. Karp.
\newblock Reducibility among combinatorial problems.
\newblock In {\em Proceedings of a symposium on the Complexity of Computer
  Computations, held March 20-22, 1972, at the {IBM} Thomas J. Watson Research
  Center, Yorktown Heights, New York.}, pages 85--103, 1972.

\bibitem{K61}
P.W. Kasteleyn.
\newblock The statistics of dimers on a lattice, i., the number of dimer
  arrangements on a quadratic lattice.
\newblock {\em Physica}, 27:166--1672, 1961.

\bibitem{K08a}
Andreas Knoblauch.
\newblock Closed-form expressions for the moments of the binomial probability
  distribution.
\newblock {\em SIAM Journal of Applied Mathematics}, 69(1):197--204, 2008.

\bibitem{KK07}
Michihiro Kuramochi and George Karypis.
\newblock Discovering frequent geometric subgraphs.
\newblock {\em Inf. Syst.}, 32(8):1101--1120, 2007.

\bibitem{LP86}
L\'{a}szl\'{o} Lov\'{a}sz and Michael~D. Plummer.
\newblock {\em Matching Theory}.
\newblock MS Chelsea Publishing.

\bibitem{LV97}
Michael Luby and Eric Vigoda.
\newblock Approximately counting up to four (extended abstract).
\newblock In {\em Proceedings of the Twenty-ninth Annual ACM Symposium on
  Theory of Computing}, STOC '97, pages 682--687, New York, NY, USA, 1997. ACM.

\bibitem{PCJ06}
Natasa Przulj, Derek~G. Corneil, and Igor Jurisica.
\newblock Efficient estimation of graphlet frequency distributions in
  protein-protein interaction networks.
\newblock {\em Bioinformatics}, 22(8):974--980, 2006.

\bibitem{R94}
Lars~Eilstrup Rasmussen.
\newblock Approximating the permanent: {A} simple approach.
\newblock {\em Random Struct. Algorithms}, 5(2):349--362, 1994.

\bibitem{R97}
Lars~Eilstrup Rasmussen.
\newblock Approximately counting cliques.
\newblock {\em Random Struct. Algorithms}, 11(4):395--411, 1997.

\bibitem{R00}
Oliver Riordan.
\newblock Spanning subgraphs of random graphs.
\newblock {\em Combinatorics, Probability {\&} Computing}, 9(2):125--148, 2000.

\bibitem{T91}
Seinosuke Toda.
\newblock {PP} is as hard as the polynomial-time hierarchy.
\newblock {\em {SIAM} J. Comput.}, 20(5):865--877, 1991.

\bibitem{U76}
Julian~R. Ullmann.
\newblock An algorithm for subgraph isomorphism.
\newblock {\em J. {ACM}}, 23(1):31--42, 1976.

\bibitem{V79}
Leslie~G. Valiant.
\newblock The complexity of computing the permanent.
\newblock {\em Theor. Comput. Sci.}, 8:189--201, 1979.

\bibitem{W06}
Dror Weitz.
\newblock Counting independent sets up to the tree threshold.
\newblock In {\em Proceedings of the Thirty-eighth Annual ACM Symposium on
  Theory of Computing}, STOC '06, pages 140--149, New York, NY, USA, 2006. ACM.

\end{thebibliography}
%\end{spacing}

\end{document}